\documentclass[12pt]{article}

\usepackage{amsmath,amssymb,amsthm}
\usepackage{graphicx}
\usepackage{fullpage}
\usepackage{enumerate}

\newtheorem{theorem}{Theorem}[section]
\newtheorem{proposition}[theorem]{Proposition}

\theoremstyle{remark}

\newtheorem{example}[theorem]{Example}
\newtheorem{remark}[theorem]{Remark}

\def\mbf#1{\mathchoice{\hbox{\boldmath $\displaystyle #1$}}
        {\hbox{\boldmath $\textstyle #1$}}
        {\hbox{\boldmath $\scriptstyle #1$}}
        {\hbox{\boldmath $\scriptscriptstyle #1$}}}

\newcommand{\X}{{\mbf X}}
\newcommand{\bX}{{\mbf X}}
\newcommand{\x}{{\mbf x}}

\newcommand{\bI}{I}

\newcommand{\bY}{{\mbf Y}}

\newcommand{\bK}{{\mbf K}}

\newcommand{\Lp}[1][p]{\boldsymbol{L}^{#1}}
\renewcommand{\kappa}{\varkappa}

\newcommand{\R}{{\mathbb R}}

\newcommand{\E}{{\mathbf E}}
\newcommand{\sA}{\mathcal{A}}

\newcommand{\sZ}{\mathcal{Z}}

\renewcommand{\P}{\mathbf{P}}
\newcommand{\Prob}[1]{\P\{#1\}}
\newcommand{\one}{\mathbf{1}}
\newcommand{\eps}{\varepsilon}
\newcommand{\rhos}{\mathsf{R}}

\newcommand{\salg}{\mathfrak{F}}

\newcommand{\risk}{r}
\newcommand{\vecrisk}{\mathbf{r}}

\DeclareMathOperator{\cl}{cl}
\DeclareMathOperator{\conv}{\overline{conv}}
\DeclareMathOperator{\essinf}{essinf}

\newlength{\querylen}
\usepackage{fancybox}

\begin{document}

\title{Multivariate risk measures in the non-convex setting}

\author{Andreas Haier and Ilya Molchanov}

\date{\today}

\maketitle

\begin{abstract}
  The family of admissible positions in a transaction costs model is a
  random closed set, which is convex in case of proportional
  transaction costs. However, the convexity fails, e.g. in case of
  fixed transaction costs or when only a finite number of transfers
  are possible. The paper presents an approach to measure risks of
  such positions based on the idea of considering all selections of
  the portfolio and checking if one of them is acceptable. Properties
  and basic examples of risk measures of non-convex portfolios are
  presented.
\end{abstract}

\section{Introduction}
\label{sec:introduction}

Multivariate financial positions (portfolios) are usually described by
vectors in Euclidean space. However, if one aims to take into account
possible exchanges between the components of the portfolio, it is
necessary to consider the whole set of points in space that may be
attained from the original position by allowed exchanges. In other
words, considering a multiasset portfolio is indispensable from
specifying which transactions may be applied to its components. For
instance, if all components of the portfolio $C=(C^{(1)},\dots,C^{(d)})$
represent cash amounts in the same currency and transfers between the
components are unrestricted with short-selling permitted, then the
attainable positions are all random vectors such that the sum of their
components equals the sum of components of $C$. By allowing disposal
of assets (e.g., in the form of consumption), we arrive at the
half-space
\begin{displaymath}
  \Big\{\x\in\R^d:\; \sum_{i=1}^d x^{(i)}\leq \sum_{i=1}^d C^{(i)}\Big\}. 
\end{displaymath}
In this case and also in the presence of transaction costs not
influenced by $C$, the attainable positions are points from $C+\bK$,
where $\bK$ is the set of portfolios available at price zero, see
\cite{kab:saf09}. In other situations, possible attainable positions
may depend on $C$ in a non-linear way, for instance, when components
represent capitals of members of a group and admissible transfers
satisfy further restrictions, e.g., requiring that they do not cause
insolvency of an otherwise solvent agent, see \cite{haier:mol:sch15}.

In view of the above reasons, it is natural to represent multiasset
portfolios as random closed sets. Recall that a \emph{random closed set}
$\bX$ is a measurable map from a probability space $(\Omega,\salg,\P)$ to the
space of closed sets in $\R^d$ equipped with the $\sigma$-algebra generated
by the Fell topology. In other words, the measurability of $\bX$ means
that $\{\omega:\; \X(\omega)\cap K\neq\emptyset\}\in\salg$ for all
compact sets $K$ in $\R^d$, see \cite[Sec.~1.1.1]{mo1}. 

A random closed set $\bX$ is said to be \emph{lower} if almost all its
realisations are lower sets, that is, for almost all $\omega$,
$x\in\bX(\omega)$ and $y\leq x$ coordinatewisely imply that
$y\in\bX(\omega)$. A random closed set is said to be \emph{convex} if
almost all its realisations are convex. If $\bX$ is a random closed
set, then its closed convex hull $\conv(\bX)$ is also a random closed
set, see \cite[Th.~1.3.25]{mo1}. 

For $p\in[1,\infty]$, denote
by $\Lp(\bX)$ the family of $p$-integrable (essentially bounded if
$p=\infty$) random vectors $\xi$ such that $\xi\in\bX$ a.s.; such
random vectors are called \emph{$p$-integrable selections} of
$\bX$. Furthermore, $\Lp[0](\bX)$ is the family of all selections of
$\bX$; this family is not empty if $\bX$ is a.s.\ non-empty, see
\cite[Th.~1.4.1]{mo1}.
A random closed set $\bX$ is called \emph{$p$-integrable} if it admits at
least one $p$-integrable selection; it is called \emph{$p$-integrably
bounded} if 
\begin{displaymath}
  \|\bX\|=\sup\{\|x\|:\; x\in\bX\}
\end{displaymath}
is a $p$-integrable random variable for $p\in[1,\infty)$. The random
closed set $\bX$ is said to be essentially bounded if $\|\bX\|$ is
a.s.\ bounded by a constant. 

If $\bX$ is integrable (that is, 
$1$-integrable), its \emph{selection expectation} is defined by 
\begin{equation}
  \label{eq:5}
  \E\bX=\cl\{\E\xi:\; \xi\in\Lp[1](\bX)\},
\end{equation}
where $\cl(\cdot)$ denotes the topological closure in $\R^d$.
The closed \emph{Minkowski sum}
\begin{displaymath}
  \bX+\bY=\cl\{x+y:\; x\in\bX,y\in\bY\}
\end{displaymath}
of two random closed sets $\bX$ and $\bY$
is also a random closed set. Note that 
\begin{displaymath}
  -\bX=\{-x:\; x\in\bX\}
\end{displaymath}
denotes the reflection of $\bX$ with respect to the origin; this is
not the inverse operation to the addition.  We refer to \cite{mo1} for
further material concerning random closed sets.

The paper is organised as follows. In
Section~\ref{sec:select-risk-meas} we introduce the \emph{selection
  risk measure} of possibly non-convex random lower closed sets,
thereby generalising the setting of \cite{haier:mol:sch15} and
\cite{cas:mol14}.  Due to the non-convexity, it is not possible to
assess the risk by working with half-spaces containing the portfolio,
as it is the case in \cite{ham:hey10,ham:hey:rud11}.  In
Section~\ref{sec:fixed-points-expect} we discuss two basic set-valued
risk measures, one based on considering the fixed points of set-valued
portfolio, the other is given by the selection expectation of
$-\bX$. These two cases correspond to taking the negative essential
infimum and the negative expectation as the underlying numerical risk
measures. Section~\ref{sec:conv-law-invar} explores the cases when the
selection risk measure takes convex values and is law invariant. The
important case of fixed transaction costs is
considered in Section~\ref{sec:fixed-trans-costs}. Finally,
Section~\ref{sec:finite-trans-sets} deals with the case of only a
finite set of admissible transactions.

\section{Selection risk measure of non-convex portfolios}
\label{sec:select-risk-meas}

\subsection{Definition}
\label{sec:definition}

Fix $p\in\{0\}\cup [1,\infty]$ and a vector
$\vecrisk(\xi)=(\risk_1(\xi^{(1)}),\dots,\risk_d(\xi^{(d)}))$ of
monetary $\Lp$-risk measures applied to components of a $p$-integrable
random vector $\xi=(\xi^{(1)},\dots,\xi^{(d)})$. We refer to
\cite{delb12} and \cite{foel:sch04} for the facts concerning risk
measures for random variables.  Assume that $\vecrisk(0)=0$ and that
all components of $\vecrisk$ are finite on $p$-integrable random
variables. When saying that $\vecrisk$ is coherent or convex, we mean
that all its components are coherent or convex. The convexity or
coherency properties will be imposed only when necessary and will be
explicitly mentioned.

In many cases below, we consider the following basic numerical risk
measures. 
\begin{enumerate}[1.]
\item The negative essential infimum $\risk(\xi)=-\essinf \xi$, which
  is an $\Lp[\infty]$-risk measure.
\item The negative expectation $\risk(\xi)=-\E\xi$, an $\Lp[1]$-risk
  measure.
\item The Average Value-at-Risk (or Expected Shortfall in the
  non-atomic case) 
  \begin{displaymath}
    \risk(\xi)=- \frac{1}{\alpha} \int_0^{\alpha} F_\xi^{-1}(t) dt.
  \end{displaymath}
  at level $\alpha\in(0,1]$ for $\xi\in\Lp[1](\R)$, where $F_\xi$ is
  the cumulative distribution function of $\xi$ and $F_\xi^{-1}$ is
  the quantile function. 
\item The distortion risk measure
  \begin{equation}
    \label{eq:77}
    \risk(\xi)=-\int_0^1 F_{\xi}^{-1}(t)d\tilde{g}(t)
  \end{equation}
  for $\xi\in\Lp(\R)$, where $g:[0,1]\mapsto [0,1]$ is a (concave)
  distortion function, $\tilde{g}(t)=1-g(1-t)$ is the dual distortion
  function, and $p$ is chosen to ensure that the integral is finite.
\end{enumerate}

The \emph{selection risk measure} of a $p$-integrable lower random closed set
$\bX$ is defined as
\begin{equation}
  \label{eq:1}
  \rhos(\bX)=\cl \bigcup_{\xi\in\Lp(\bX)} (\vecrisk(\xi)+\R_+^d),
\end{equation}
where the union is taken over all $p$-integrable selections of $\bX$.
Thus, $x\in\rhos(\bX)$ if and only if $\liminf \vecrisk(\xi_n)\leq x$
for $\xi_n\in\Lp(\bX)$, $n\geq1$. The inequalities between vectors are
always coordinatewise and the lower limit is also taken
coordinatewisely. The selection risk measure takes values being upper
sets, and \eqref{eq:1} can be seen as the primal representation of
$\rhos(\bX)$. A dual representation is not feasible without imposing
convexity on $\bX$.

A random set $\bX$ is said to be acceptable if $0\in\rhos(\bX)$. In
other words, $\bX$ is acceptable if $\bX$ contains a sequence
of selections whose risk converges to zero. The monetary property of
$\vecrisk$ yields that $\rhos(\bX)$ is the set of all $x\in\R^d$ such
that $\bX+x$ is acceptable, that is,
\begin{displaymath}
  \rhos(\bX)=\{x:\; \rhos(\bX+x)\ni 0\}.
\end{displaymath}

\subsection{Properties of the selection risk measure}
\label{sec:prop-select-risk}

The selection risk measure was introduced in \cite{cas:mol14} for
convex $\bX$ and coherent $\vecrisk$. Some of its properties for
non-convex $\bX$ and general monetary $\vecrisk$ are easy-to-show
replica of those known in the convex coherent setting adopted in
\cite{cas:mol14}. 

\begin{theorem}
  \label{thr:general}
  The selection risk measure satisfies the following properties for
  $p$-integrable random lower closed sets $\bX$ and $\bY$. 
  \begin{enumerate}[i)]
  \item Monotonicity, that is, $\rhos(\bX)\subseteq \rhos(\bY)$ if
    $\bX\subseteq \bY$ a.s.
  \item Cash-invariance, that is, $\rhos(\bX+a)=\rhos(\bX)-a$ for
    all deterministic $a\in\R^d$.
  \item If $\vecrisk$ is homogeneous, then $\rhos$ is homogeneous,
    that is, $\rhos(c\bX)=c\rhos(\bX)$ for all deterministic $c>0$.
  \item If $\vecrisk$ is convex, then $\rhos$ is convex, that is, 
    \begin{equation}
      \label{eq:3}
      \rhos(\lambda\bX+(1-\lambda)\bY)
      \supseteq \lambda\rhos(\bX)+(1-\lambda)\rhos(\bY)
    \end{equation}
    for all deterministic $\lambda\in[0,1]$. 
  \end{enumerate}
\end{theorem}
\begin{proof}
  We prove only the convexity, the rest is straightforward. All
  elements of the set on the right-hand side of \eqref{eq:3} are
  coordinatewisely larger than or equal to 
  \begin{displaymath}
    \liminf \big(\lambda\vecrisk(\xi_n)+(1-\lambda)\vecrisk(\eta_n)\big)
  \end{displaymath}
  for $\xi_n\in\Lp(\bX)$ and $\eta_n\in\Lp(\bY)$, $n\geq1$. Then it
  suffices to note that this convex combination of risks of $\xi$ and
  $\eta$ dominates $\vecrisk(\lambda\xi_n+(1-\lambda)\eta_n)$, which is an
  element of the left-hand side of \eqref{eq:3}.
\end{proof}

The monotonicity property of $\vecrisk$ yields that
$\rhos(C+\R_-^d)=\vecrisk(C)+\R_+^d$ for $C\in\Lp(\R^d)$. The
selection risk measure is said to be \emph{coherent} if it is
homogeneous and convex; this is the case if $\vecrisk$ has all
coherent components. If $\vecrisk$ is coherent, $C$ is a
$p$-integrable random vector, and $\bX$ is a $p$-integrable random
lower closed set, then
\begin{equation}
  \label{eq:10}
  \rhos(C+\bX)\supseteq \vecrisk(C)+\rhos(\bX).
\end{equation}
This is easily seen from \eqref{eq:3} choosing $\lambda=1/2$,
$\bY=C+\R_-^d$, and using the homogeneity of $\vecrisk$. Note that the
equality in \eqref{eq:10} is not guaranteed even if $\bX$ is a
deterministic set. Still, in this case, it provides a useful
acceptability condition: $C+\bX$ is acceptable if
$\vecrisk(C)+\rhos(\bX)\ni 0$. 

A general set-valued function (not necessarily
constructed using selections) defined for $p$-integrable random sets
is said to be monotonic, cash invariant, homogeneous or convex if it
satisfies the corresponding properties from
Theorem~\ref{thr:general}. The set-valued (selection) risk measure is
called \emph{law invariant} if its values on identically distributed
random sets coincide.

\subsection{Choice of selections}
\label{sec:choice-selections}

The definition of the selection risk measure involves taking union
over all $p$-integrable selections of $\bX$. This family may be very
rich even for simple random closed sets. In the following, we discuss
general approaches suitable to reduce the family of selections needed
to determine the selection risk measure. 

With a lower closed set $F$ we associate the set $\partial^+F$ of its
\emph{Pareto optimal} points, that is, the set of points $x\in F$ such
that $y\geq x$ for $y\in F$ is only possible if $y=x$. If $\bX$ is a
random lower closed convex set, then the set $\partial^+\bX$ of Pareto
optimal points of $\bX$ is a random closed set, see
\cite[Lemma~3.1]{haier:mol:sch15}. In the non-convex case, the cited
result establishes that $\partial^+\bX$ is graph measurable, so that
its closure $\cl \partial^+\bX$ is a random closed set, see
\cite[Prop.~2.6]{lep:mol17}. 
If $\partial^+\bX$ is closed and $p$-integrable, then it is possible
to reduce the union in \eqref{eq:1} to selections of
$\partial^+\bX$.

A lower random closed set $\bX$ is said to be \emph{quasi-bounded} if
$\partial^+\bX$ is essentially bounded; $\bX$ is $p$-integrably
quasi-bounded if $\|\partial^+\bX\|$ is $p$-integrable.

Consider
\begin{equation}
  \label{eq:8}
  \bX=F_1\cup\cdots \cup F_m,
\end{equation}
where $F_1,\dots,F_m$ are deterministic lower \emph{convex closed
  cones}. For the following result, assume that $\vecrisk$ is convex
law invariant, and the probability space is non-atomic. In this case,
$\vecrisk$ satisfies the dilatation monotonicity property, that is,
$\vecrisk(\xi)$ dominates coordinatewisely the risk of a conditional
expectation of $\xi$, see \cite[Cor.~4.59]{foel:sch04} and
\cite{leit04}.


\begin{proposition}
  \label{prop:det-set}
  If $\bX$ is a deterministic set 
  given by \eqref{eq:8}, then it is possible to reduce the
  union in \eqref{eq:1} to selections $\xi=\sum_{i=1}^m x_i\one_{A_i}$
  for deterministic $x_i\in F_i$, $i=1,\dots,m$, and partitions
  $\sA=\{A_1,\dots,A_m\}$ of the probability space.
\end{proposition}
\begin{proof}
  Consider $\xi=\sum \eta_i\one_{A_i}$ for $\eta_i\in\Lp(F_i)$,
  $i=1,\dots,m$. By the dilatation monotonicity, $\vecrisk(\xi)$
  dominates the risk of the conditional expectation of $\xi$ given
  $\sA$. Thus, it is possible to replace $\eta_i$ by its conditional
  expectation, which is also a point in $F_i$.
\end{proof}

In the convex setting, if $\bX$ is the sum of $C$ and a convex closed
set $F$, then the union in \eqref{eq:1} can be reduced to the
selections that are measurable with respect to the $\sigma$-algebra
generated by $C$. 

\section{Fixed points and the expectation}
\label{sec:fixed-points-expect}

For a random closed set $\bX$,
\begin{displaymath}
  F_\bX=\{x:\; \Prob{x\in\bX}=1\}
\end{displaymath}
denotes the set of its \emph{fixed points}. The set $F_\bX$ is a lower
closed set if $\bX$ is a lower closed set, it is convex if $\bX$ is
convex.

\begin{proposition}
  \label{prop:ess-inf}
  Let $\bX$ be a $p$-integrable random lower closed set.  For the
  selection risk measure generated by any monetary risk measure
  $\vecrisk$, we have
  \begin{equation}
    \label{eq:6}
    -F_\bX\subseteq \rhos(\bX). 
  \end{equation}
  If all components of $\vecrisk$
  are the negative of the essential infimum, then $\rhos(\bX)$ equals
  the set of fixed points of $-\bX$.
\end{proposition}
\begin{proof}
  By taking constant selections $\xi=x\in F_\bX$ in \eqref{eq:1} and
  using the fact that $\vecrisk(x)=-x$, we see that \eqref{eq:6}
  holds.  

  If $\Lp[\infty](\bX)\neq\emptyset$, then $F_\bX\neq\emptyset$, since
  $\bX$ is a lower set.  Choosing $\vecrisk$ with all components
  being negative of the essential infima, it is easily seen that $\bX$
  is acceptable if it has a selection with all a.s. non-negative
  components. In this case, $0\in\bX$ a.s., whence $0\in F_\bX$. Note
  also that $F_{-\bX}=-F_\bX$.
\end{proof}

The set of fixed points is a coherent selection risk measure, which is law
invariant and not necessarily convex-valued.

\begin{example}
  \label{ex:fx-non-convex}
  The convex hull of $F_\bX$ is a (possibly, strict) subset of the set
  of fixed points of $\conv(\bX)$.  Let $\bX$ be a random set in
  $\R^2$ which equally likely
  take values $\{(-a,a),(a,-a)\}+\R_-^2$ and
  $\{(-b,b),(b,-b)\}+\R_-^2$ for $0<a<b$. Then
  $F_\bX=\{(-b,a),(a,-b)\}+\R_-^2$, while the set of fixed points of
  $\conv(\bX)$ is the sum of the segment with end-points
  $(-a,a),(a,-a)$ and $\R_-^2$.
\end{example}

\begin{example}
  The set of fixed points appears also in the following context. Let
  $\Omega=\{\omega_1,\dots,\omega_n\}$ be a finite probability space,
  and let all components of $\vecrisk$ be the Average Value-at-Risk at
  level $\alpha\leq\P(\{\omega_i\})$, $i=1,\dots,n$. Then
  \begin{displaymath}
    \rhos(\X)=-F_\bX=-\bigcap_{i=1}^{n} \bX(\omega_i).
  \end{displaymath}
  Indeed, since $\P(\{\omega_i\})\geq \alpha$ for all $i$, we have
  $\risk(\xi^{(j)})=-\min\{\xi^{(j)}(\omega_i),i=1,\dots,n\}$ for any
  $\xi\in\Lp[1](\bX)$. Because each $\bX(\omega_i)$ is a lower set, we
  have $-\vecrisk(\xi)\in \bX(\omega_i)$ for all $i$. To show the
  reverse inclusion, assume that $x\in F_\bX$. Then $\xi=x$ is a
  deterministic selection of $\bX$, whence $-x=\vecrisk(\xi)\in
  \rhos(\bX)$.
\end{example}

\medskip

If $p=1$ and $\vecrisk(\xi)=-\E\xi$ is the negative expectation of
$\xi$, then $\rhos(\bX)$ becomes the selection expectation of $(-\bX)$.
Note that $\rhos(\bX)=-\E\bX$ is a coherent selection risk measure,
which is law invariant on convex random sets, but may be not law
invariant on non-convex ones. Indeed, if the non-convex deterministic
set $F$ is considered a random closed set defined on the trivial
probability space, then $\E F=F$, while $\E F=\conv(F)$ if the
underlying probability space is non-atomic, see \cite[Th.~2.1.26]{mo1}.

It might be tempting to define a set-valued risk measure by taking
intersection of expected random sets with respect to varying
probability measures. This would correspond to the construction of a
convex function by taking the supremum of linear ones. However, taking
expectation results in convex values for the risk measure if the
probability space is non-atomic; otherwise, it depends on the atomic
structure of the space. Furthermore, even in the convex setting, such
a construction might not correspond to the existence of an acceptable
selection from $\bX$, as the following remark shows.

\begin{remark}
  For any family $\sZ\subset\Lp[q](\R_+^d)$ such that $\E\zeta=1$ for
  all $\zeta\in\sZ$, define
  \begin{equation}
    \label{eq:4}
    \rhos_\sZ(\bX)=\bigcap_{\zeta\in\sZ} \E(-\zeta\bX),
  \end{equation}
  where $\zeta\bX=\{\zeta x:\; x\in\bX\}$.  Note that we use
  vector notation, e.g., $\E\zeta=1$ means that all components
  of $\zeta$ have mean $1$, and
  \begin{displaymath}
    \zeta\xi=(\zeta^{(1)}\xi^{(1)},\dots,\zeta^{(d)}\xi^{(d)})
  \end{displaymath}
  is the coordinatewise product of $\zeta$ and $\xi$.  The so defined
  $\rhos_\sZ(\cdot)$ satisfies all properties from
  Theorem~\ref{thr:general}. However, $\rhos_\sZ$ in general is not a
  selection risk measure. Indeed, by letting $\bX=\xi+\R_-^d$, we see
  that the corresponding coherent vector-valued risk measure is given by
  \begin{displaymath}
    \vecrisk(\xi)=\sup_{\zeta\in\sZ} \E(-\zeta\xi),\quad
    \xi\in\Lp(\R^d). 
  \end{displaymath}
  Assume that $\partial^+\bX$ is $p$-integrably bounded, so that
  $\E(\zeta\bX)$ is closed for all $\zeta\in\Lp[q]$. Then
  $0\in\rhos_\sZ(\bX)$ if and only if $0\in\E(-\zeta\bX)$ for all
  $\zeta\in\sZ$, equivalently, for each $\zeta\in\sZ$ there is
  $\xi_\zeta\in\Lp(\partial^+\bX)$ such that $\E(-\zeta\xi_\zeta)\leq
  0$. Since these selections $\xi_\zeta$ may be different for
  different $\zeta$, we cannot infer that $\bX$ is acceptable with
  respect to a selection risk measure. Indeed, the acceptability of
  $\bX$ requires the existence of a \emph{single} selection $\xi\in\Lp(\bX)$
  such that $\E(-\zeta\xi)\leq 0$ for all $\zeta$.
  Thus, $\rhos_\sZ$ is an example of a coherent set-valued risk
  measure, which, however, is not necessarily a selection one. The
  acceptability of $\bX$ under $\rhos_\sZ$ does not guarantee the
  existence of an acceptable selection of $\bX$. Furthermore, this
  risk measure does not distinguish between $\bX$ and its convex
  hull. 
\end{remark}

\section{Convexity and law invariance}
\label{sec:conv-law-invar}

The monotonicity property yields that $\rhos(\bX)$ is a subset of
$\rhos(\conv(\bX))$.  It is well known that the
selection expectation of an integrable random closed set is convex if
the underlying probability space is non-atomic, see \cite[Th.~2.1.26]{mo1}. This
result follows from Lyapunov's theorem on ranges of vector-valued
measures. The same holds for selection risk measures of convex random
sets, if the underlying risk measure $\vecrisk$ is convex, see
\cite[Th.~3.4]{cas:mol14}.  This is however not the case for
non-convex arguments, see Example~\ref{ex:fx-non-convex} and
Section~\ref{sec:fixed-trans-costs-1}.

Still, in some cases $\rhos(\bX)$ is convex even for non-convex $\bX$.
Assume that $p\in[1,\infty]$, and the components of
$\vecrisk=(\risk_1,\dots,\risk_d)$ are $\sigma(\Lp,\Lp[q])$-lower
semicontinuous convex risk measures, so that
\begin{equation}
  \label{eq:2}
  \risk_i(\xi)=\sup_{\zeta\in\Lp[q](\R_+), \E\zeta=1}
  \Big(\E(-\zeta\xi)-\alpha_i(\zeta)\Big), \quad \xi\in\Lp(\R),
  \quad i=1,\dots,d,
\end{equation}
where 
$\alpha_i:\Lp[q](\R_+)\mapsto(-\infty,\infty]$, $i=1,\dots,d$, are the
penalty functions corresponding to the components of $\vecrisk$. The
following result generalises Lyapunov's theorem in the sublinear
setting, see also \cite{sag09}. 

\begin{theorem}
  \label{thr:convexify}
  Let $(\Omega,\salg,\P)$ be a non-atomic probability space, and let
  the components of $\vecrisk$ admit representation \eqref{eq:2} with
  $\alpha_1(\zeta),\dots,\alpha_d(\zeta)$ being all infinite unless
  $\zeta$ belongs to a finite family from $\Lp[q](\R)$. Then
  $\rhos(\X)$ is convex.
\end{theorem}
\begin{proof}
  We need to show that for two selections $\xi',\xi''\in\Lp(\bX)$ and
  $\lambda \in [0,1]$, there is $\xi\in\Lp(\bX)$ such that
  $\vecrisk(\xi)\leq\lambda\vecrisk(\xi')+(1-\lambda)\vecrisk(\xi'')$. In
  view of the convexity of $\vecrisk$, it suffices to ensure that
  \begin{displaymath}
    \vecrisk(\xi)\leq \vecrisk(\lambda\xi'+(1-\lambda)\xi'').
  \end{displaymath}

  Assume that all components of $\alpha(\zeta)$ are infinite for
  $\zeta$ outside a finite set
  $\sZ=\{\zeta_1,\dots,\zeta_m\}$. Consider the mapping which assigns
  to each measurable subset $A\subseteq \Omega$ the vector
  \begin{displaymath}
    \upsilon(A)=(\E(-\one_A\zeta_1\xi'),\dots,\E(-\one_A\zeta_m\xi'),
    \E(-\one_A\zeta_1\xi''),\dots,\E(-\one_A\zeta_m\xi''))\in \R^{2md}.
  \end{displaymath}
  It is easily verified that this map is a vector-valued measure. By
  Lyapunov's theorem, its image is closed convex, hence there is a
  measurable subset $A\subseteq \Omega$ such that
  \begin{displaymath}
    \upsilon(A)=\lambda
    \upsilon(\Omega)+(1-\lambda)\upsilon(\emptyset)
    =\lambda\upsilon(\Omega).
  \end{displaymath}
  Then $\E(-\one_A\zeta_i\xi')=\lambda\E(-\zeta_i\xi')$ and
  $\E(-\one_A\zeta_i\xi'')=\lambda\E(-\zeta_i\xi'')$ for all $i$. Hence,
  \begin{displaymath}
    \E(-\lambda\zeta_i\xi'-(1-\lambda)\zeta_i\xi'')
    =\E(-\zeta_i(\xi''+\one_A(\xi'-\xi'')))=\E(-\zeta_i\xi),
  \end{displaymath}
  where $\xi=\xi'\one_A+\xi''\one_{A^c}$ is a selection of $\bX$. Therefore,
  \begin{displaymath}
    \E(-\zeta_i\xi)-\alpha(\zeta_i)
    =\E(-\zeta_i(\lambda\xi'+(1-\lambda)\xi''))-\alpha(\zeta_i)
    \leq \vecrisk(\lambda\xi'+(1-\lambda)\xi'')
  \end{displaymath}
  for all $i$, so $\xi$ is indeed the required selection. 
\end{proof}

\begin{remark}
  For a \emph{deterministic} lower closed set $F$, the selection risk
  measure $\rhos(F)$ is not always equal to $(-F)$. For instance, this
  is not the case in the framework of Theorem~\ref{thr:convexify}, or
  in the context of fixed transaction costs in
  Section~\ref{sec:fixed-trans-costs-1}.  The set $F$ is said to be
  $\vecrisk$-convex (or risk-convex for $\vecrisk$), if with any $x_1,x_2\in F$ and
  any $A\subseteq \Omega$ we also have $-\vecrisk(\one_A x_1+
  \one_{A^c} x_2)\in F$. Then $F$ is $\vecrisk$-convex if and only if
  $\rhos(F)=-F$. It is easy to see that the intersection of risk
  convex sets is also risk convex.  If $\vecrisk$ is the negative
  expectation and the probability space is non-atomic, the risk
  convexity corresponds to the usual notion of convexity. If
  $\vecrisk$ is the negative essential infimum, each lower set is risk
  convex.
\end{remark}

\begin{remark}
  Consider $\bX=\{\xi,\eta\}+\R_-^d$. Then $\rhos(\bX)$ is convex if
  and only if, for each $t\in(0,1)$, there exists $A\in\salg$ such
  that 
  \begin{displaymath}
    t\vecrisk(\xi)+(1-t)\vecrisk(\eta)\geq
    \vecrisk(\xi\one_A+\eta\one_{A^c}). 
  \end{displaymath}
\end{remark}

\bigskip

The families of selections of random sets are not necessarily law
invariant, i.e. they can differ for two random sets having the same
distribution, see \cite[Sec.~1.4.1]{mo1}. This could result in the selection risk
measure $\rhos$ not being law invariant. Still, the law invariance of
$\vecrisk$ yields the law invariance of the selection risk measure for
convex $\bX$, see \cite{cas:mol14}. Below we consider the case of a
possibly non-convex $\bX$.

The risk measure $\vecrisk$ is said to be \emph{Lebesgue continuous} if it is
continuous on a.s. convergent uniformly $p$-integrably bounded
sequences of random vectors. 

\begin{theorem}
  \label{thr:law-inv}
  Assume that the probability space is non-atomic and that $\vecrisk$
  is a Lebesgue continuous risk measure.  Then the selection risk
  measure $\rhos(\bX)$ is law invariant on $p$-integrably
  quasi-bounded portfolios.
\end{theorem} 
\begin{proof}
  Let $\bX$ and $\bX'$ share the same distribution, so that the
  corresponding closures $\bY=\cl\partial^+\bX$ and
  $\bY'=\cl\partial^+\bX'$ of their Pareto optimal points are
  $p$-integrably bounded and share the same distribution. By the
  Lebesgue property and the $p$-integrable boundedness of
  $\bY$ and $\bY'$, it is possible to take the
  union in \eqref{eq:1} over $p$-integrable selections of
  $\bY$ and $\bY'$ respectively.
  
  Let $x\in\vecrisk(\xi)+\R_+^d$ for $\xi\in\Lp(\bY)$. Since
  the weak closures of $\Lp[0](\bY)$ and
  $\Lp[0](\bY')$ coincide (see \cite[Th.~1.4.3]{mo1}), there
  is a sequence $\eta_n\in\Lp(\bY')$ converging weakly to
  $\xi$. Then $\|\eta_n\|\leq \|\bY'\|$, and the latter
  random variable is integrable. Thus, $\{\eta_n,n\geq1\}$ is
  relatively compact in $\Lp[1](\R^d)$. By passing to a subsequence,
  it is possible to assume that $\eta_{n_k}\to\xi$ almost surely.

  The Lebesgue continuity property yields that
  $\vecrisk(\eta_{n_k})\to \vecrisk(\xi)$.  Thus,
  $\vecrisk(\xi)\in\rhos(\bY')$, since the latter set is
  closed. Finally, $x\in\rhos(\bX')$, since the latter set is upper.
\end{proof}

It is known that each $\Lp$-risk measure with finite values and
$p\in[1,\infty)$ is Lebesgue continuous, see \cite{kain:rues09}.  For
$p=\infty$, \cite[Thms.~2.4, 5.2]{jouin:sch:touz06} provide equivalent
formulations of the Lebesgue continuity property for convex risk
measures. We give below another criterion.

\begin{proposition}
  \label{prop:leb-p}
  Assume that $\vecrisk$ is a coherent $\Lp[\infty]$-risk measure such that 
  \begin{displaymath}
    \vecrisk(\xi)=\sup_{\zeta\in\sZ} \E(-\zeta\xi),
  \end{displaymath}
  where $\sZ$ is a uniformly integrable subset of
  $\Lp[1](\R_+^d)$. 
  Then $\vecrisk$ is Lebesgue continuous.
\end{proposition}
\begin{proof}
  Assume that $\xi_n\to\xi$ a.s. and $\|\xi_n\|\leq c$ a.s. for all
  $n$ and $c>0$.  By Egorov's theorem, for each $\eps>0$, there is an
  event $A$ of probability at most $\eps$ such that $\xi_n\to\xi$
  uniformly on the complement $A^c$ of $A$.

  Using the fact that the absolute value of the difference of two
  suprema is bounded by the suprema of the absolute values of the
  differences, we have 
  \begin{displaymath}
    \|\vecrisk(\xi_n)-\vecrisk(\xi)\|
    \leq \sup_{\zeta\in\sZ} \|\E(-\zeta(\xi_n-\xi)\|
    \leq \sup_{\zeta\in\sZ}\E\|\zeta\| \sup_{\omega\notin
      A} \|\xi_n(\omega)-\xi(\omega)\|+2c\|\E(-\zeta\one_A)\|.
  \end{displaymath}
  The first term on the right-hand side converges to zero by the
  uniform convergence on $A^c$, while the second converges to zero by
  the uniform integrability of $\sZ$.
\end{proof}

\section{Fixed transaction costs}
\label{sec:fixed-trans-costs}

\subsection{Bounds on the selection risk measure}
\label{sec:general-properties}

Assume that the components of $C$ represent the same currency and
transfers are not restricted, but whenever a transfer is made, a fixed
cost $\kappa>0$ is incurred. If $C$ is the capital position, then the
corresponding set of attainable positions is given by
$\bX=C+\bI_\kappa$ with non-convex set
\begin{displaymath}
  \bI_\kappa=\R_-^d\cup H_{-\kappa}
\end{displaymath}
of portfolios available at price zero, where
\begin{displaymath}
  H_t=\{x\in\R^d:\; \sum_{i=1}^d x_i\leq t\}, \quad t\in\R.
\end{displaymath}
The following bounds for the selection risk measure of $C+\bI_\kappa$ are
straightforward.

\begin{proposition}
  \label{prop:bounds}
  We have
  \begin{equation}
    \label{eq:7}
    (\vecrisk(C)-\bI_\kappa)\cup \rhos(C+H_{-\kappa})
    \; \subseteq\; \rhos(C+\bI_\kappa)\;\subseteq\; \rhos(C+H_0).
  \end{equation}
\end{proposition}
\begin{proof}
  The first inclusion follows from the fact that $C+x$ is a selection
  of $C+\bI_\kappa$ for all deterministic $x\in\bI_\kappa$ and that
  $H_{-\kappa}\subset\bI_\kappa$. The second inclusion holds, since
  $\bI_\kappa\subset H_0$.
\end{proof}

\begin{example}
  \label{ex:genC}
  The inclusion on the left-hand side of \eqref{eq:7} can be strict.
  Let $d=2$, and let $C$ be $(-1,1)$ with probability $\alpha$ and
  $(0,0)$ otherwise. For any $0\leq \beta \leq \alpha$, we can define
  a selection $\xi\in\bI_\kappa$ such that $C+\xi$ equals
  $(-\kappa,0)$ with probability $\beta$, $(-1,1)$ with probability
  $\alpha-\beta$, and $(0,0)$ with probability $1-\alpha$. Taking the
  risk measure of such selections shows that $\rhos(C+\bI_\kappa)$
  contains all points on the segments with end-points $(1,0)$ and
  $(\kappa,0)$.
\end{example}

The following result provides rather simple bounds on the selection
risk measure in case of fixed transaction costs. 

\begin{proposition}
 \begin{enumerate}[i)]
 \item If $\kappa_1\leq \kappa_2$ and $C_1\geq C_2$ componentwisely,
   then 
   \begin{displaymath}
     \rhos(C_1+\bI_{\kappa_1})\supseteq
     \rhos(C_2+\bI_{\kappa_2}).
   \end{displaymath}
 \item If $\vecrisk$ is subadditive, then
   \begin{displaymath}
     \rhos(C_1+C_2+\bI_\kappa)\supseteq \rhos(C_1+\bI_{\kappa_1}) +
     \rhos(C_2+\bI_{\kappa_2})
   \end{displaymath}
   whenever $\kappa \leq \min(\kappa_1,\kappa_2)$.
 \end{enumerate}
\end{proposition}
\begin{proof}
  i) Note that $\bI_{\kappa_1}\supseteq \bI_{\kappa_2}$ for $\kappa_1
  \leq \kappa_2$, and
  \begin{displaymath}
    C_1+\bI_{\kappa_1}\supseteq C_1+\bI_{\kappa_2}\supseteq C_2+\bI_{\kappa_2}.
  \end{displaymath}
  
  ii) follows from $\bI_{\kappa_1}+\bI_{\kappa_2} \subseteq
  \bI_\kappa$ and the monotonicity of the selection risk measure.
\end{proof}

The following result identifies the selection risk measure of $C+H_t$
in some cases in terms of the risk of the total payoff
\begin{displaymath}
  D=C^{(1)}+\cdots+C^{(d)}.
\end{displaymath}
If $\vecrisk$ is coherent with all identical components, it is easy to
see that $C+H_t$ is acceptable if and only if $D$ is acceptable. 
The following result concerns the case, when all but one components of
$\vecrisk$ are identical.

\begin{proposition}
  \label{prop:Ht}
  \begin{enumerate}[i)]
  \item If all the components of $\vecrisk$ are identical convex risk
    measures $\risk$, then
    \begin{displaymath}
      \rhos(C+H_t)=-H_{t-d\risk(D/d)}.
    \end{displaymath}
  \item If one of the components of $\vecrisk$ is the negative
    essential infimum and all others are identical convex risk
    measures $\risk$, then
    \begin{displaymath}
      \rhos(C+H_t)=-H_{t-(d-1)\risk(\frac{D}{d-1})}.
    \end{displaymath}
  \item If one of components of $\vecrisk$ is the negative expectation
    and all others are identical convex risk measures $\risk$ such
    that $\risk(\xi)\geq -\E\xi$ for all $\xi\in\Lp[1](\R)$, then
    \begin{displaymath}
      \rhos(C+H_t)=-H_{t-\E D}.
    \end{displaymath}
  \end{enumerate}
\end{proposition}
\begin{proof}
  By cash-invariance, it is possible to asssume that $t=0$. The
  statement i) is shown in \cite[Th.~5.1]{haier:mol:sch15}.
  \medskip
  
  \noindent
  ii) Assume that the first component of $\vecrisk$ is the negative
  essential infimum. Note that $0\in\rhos(C+H_0)$ if and only if there
  is a selection $\xi$ such that $\sum_{i=1}^{d} \xi^{(i)}\leq 0$,
  $C^{(1)}+\xi^{(1)}\geq 0$ a.s. and $\risk(C^{(i)}+\xi^{(i)})\leq 0$
  for $i=2,\dots, d$. By convexity and monotonicity of $\risk$,
  \begin{align*}
    \label{eq:diagonal}
    \risk\big(\frac{D}{d-1}\big)
    =\risk\Big(\frac{1}{d-1}\sum_{i=1}^{d} C^{(i)}\Big)
    &\leq\risk\Big(\frac{1}{d-1}\sum_{i=2}^{d} (C^{(i)}+\xi^{(i)})\Big)\\
    &\leq \sum_{i=2}^{d} \frac{1}{d-1}\risk(C^{(i)}+\xi^{(i)}).
  \end{align*}
  Hence, if $\risk(C^{(i)}+\xi^{(i)})\leq 0$ for all $i=2,\dots,d$, then $0\in
  -H_{-(d-1)\risk(\frac{D}{d-1})}$. On the other hand, if
  $\risk(\frac{D}{d-1})\leq 0$, then letting
  $\xi^{(1)}=-C^{(1)}$ and $\xi^{(i)}=-C^{(i)}+D/(d-1)$, $i=2,\dots,d$, yields a
  selection $\xi$ of $C+H_0$ such that $\vecrisk(C+\xi)\leq 0$.

  \medskip
  \noindent
  iii) 
  If $0\in \rhos(C+H_0)$, then there is $\xi$ such that
  $\E(C^{(1)}+\xi^{(1)})\geq 0$ and $\risk(C^{(i)}+\xi^{(i)})\leq 0$,
  $i=2,\dots,d$. Denote $\eta=D-C^{(1)}-\xi^{(1)}$. 
  Since $\sum_{i=2}^{d}\xi^{(i)}\leq -\xi^{(1)}$, we have 
  \begin{displaymath}
    \sum_{i=2}^{d} (C^{(i)}+\xi^{(i)}) 
    = D-C^{(1)}+ \sum_{i=2}^{d} \xi^{(i)}
    \leq \eta.
  \end{displaymath}
  Thus,
  \begin{equation}
    \label{eq:diagonal2}
    \risk\big(\eta/(d-1)\big)
    \leq \frac{1}{d-1} 
    \risk\big(\sum_{i=2}^{d} (C^{(i)}+\xi^{(i)})\big)
    \leq \frac{1}{d-1}\sum_{i=2}^{d} \risk(C^{(i)}+\xi^{(i)}) \leq 0. 
  \end{equation}
  Note that $\E(C^{(1)}+\xi^{(1)})\geq 0$ is equivalent to $\E\eta\leq
  \E D$. Inequality \eqref{eq:diagonal2} yields that
  $-\E\eta\leq\risk(\eta)\leq 0$. Therefore, $0\leq \E\eta\leq \E D$
  as desired.

  If $\E D\geq 0$, define a selection of $C+H_0$ by letting
  $\xi^{(1)}=-C^{(1)}+D$ and $\xi^{(i)}=-C^{(i)}$, $i=2,\dots,d$. Then
  $\E(C^{(1)}+\xi^{(1)})\geq 0$ and $C^{(i)}+\xi^{(i)}=0$ for
  $i=2,\ldots,d$, whence $0\in \rhos(C+H_0)$.
\end{proof}

\subsection{Fixed transaction costs in case $C=0$}
\label{sec:fixed-trans-costs-1}

If $C=0$, then the portfolio $\bX=C+\bI_\kappa=\bI_\kappa$ is
deterministic. However, in the non-convex case, $\rhos(\bI_\kappa)$
may be a strict superset of $(-\bI_\kappa)$. For instance, this
happens in the context of Theorem~\ref{thr:convexify} when
$\rhos(\bI_\kappa)=-H_0=\conv(-\bI_\kappa)$.

In the following assume that $\vecrisk$ is a coherent risk measure and
$d=2$. By Proposition~\ref{prop:det-set}, it suffices to consider
selections $\xi=(x,y)\one_{A}$ satisfying $x+y=-\kappa$ with
$(x,y)\notin\R_-^2$. If $x\geq 0$ and so $y\leq 0$, then
\begin{displaymath}
  \vecrisk(\xi)=(x\risk_1(\one_A),-y\risk_2(-\one_A)).
\end{displaymath}
If $x<0$, then
\begin{displaymath}
  \vecrisk(\xi)=(-x\risk_1(-\one_A),y\risk_2(\one_A)).
\end{displaymath}
Thus, the risk of $\bI_\kappa$ is determined by the set
\begin{align*}  
  B_\vecrisk=\{(\risk_1(\one_A),\risk_2(-\one_A)):\; A\in\salg\},
\end{align*}
where $\P(A)=\beta$ varies between $0$ and $1$.  Then
\begin{equation}
  \label{eq:88}
  \rhos(\bX)=\bigcup_{t\geq0, (b^{(1)},b^{(2)})\in B_\vecrisk}
  \big\{(tb^{(1)},(-\kappa-t)b^{(2)}),
  (tb^{(2)},(t-\kappa)b^{(1)})\big\}+\R_+^2. 
\end{equation}

\begin{example}
  \label{ex:c-zero}
  Let $d=2$, and let the both components of $\vecrisk=(\risk,\risk)$
  be the Average Value-at-Risk at level $\alpha$. If $\P(A)=\beta$,
  then
  \begin{displaymath}
    (\risk(\one_A),\risk(-\one_A))=
    \begin{cases}
      (0,\beta/\alpha),& \beta\leq \min\big(\alpha,1-\alpha),\\
      (0,1),& \alpha<\beta\leq 1-\alpha,\;
      \alpha\leq 1/2,\\
      (-1+(1-\beta)/\alpha,\beta/\alpha),&
      1-\alpha<\beta\leq \alpha,\;
      \alpha> 1/2,\\
      (-1+(1-\beta)/\alpha,1),& \max(\alpha,1-\alpha)<\beta\leq1.
    \end{cases}
  \end{displaymath}
  Thus, if $\alpha\leq 1/2$, then $B_\vecrisk$ is the union of two
  segments $[(0,0),(0,1)]$ and $[(0,1),(-1,1)]$ and it does not
  depend on $\alpha$. In this case, \eqref{eq:88} yields that
  $\rhos(\bI_\kappa)=-\bI_\kappa$.

  Assume now that $\alpha>1/2$. Then $B_\vecrisk$ is the line that
  joins the points $(0,0)$, $(0,1/\alpha-1)$, $(1/\alpha-2,1)$ and
  $(-1,1)$. Only the middle segment differs from the case $\alpha\leq
  1/2$. If $t>0$, then the points
  \begin{displaymath}
    \{(tb^{(1)},(t-\kappa)b^{(2)}):\; 
    (b^{(1)},b^{(2)})\in B_\vecrisk\}
  \end{displaymath}
  constitute the segment with the end-points
  $(0,(t-\kappa)(1/\alpha-1))$ and $(t(1/\alpha-2),t-\kappa)$. A
  calculation of the lower envelope of these segments yields that
  \begin{multline*}
    \rhos(\bI_\kappa)=\Big\{(-x,y):\; x\geq 0,
    y\geq\min\Big(\kappa+x,\big(\sqrt{x}+\sqrt{\kappa(1/\alpha-1)}\big)^2\Big)\Big\}\\
    \bigcup
    \Big\{(x,-y):\; y\geq 0,
    x\geq\min\Big(\kappa+y,\big(\sqrt{y}+\sqrt{\kappa(1/\alpha-1)}\big)^2\Big)\Big\}.
  \end{multline*}
  Figure~\ref{fig:1} shows the risk of $\bI_\kappa$ for $\alpha=0.75$.
  This set increases as $\alpha$ grows and becomes
  $\conv(-\bI_\kappa)$ if $\alpha=1$.

  \begin{figure}
    \centering
    \includegraphics[angle=-90,width=7cm,totalheight=8cm]{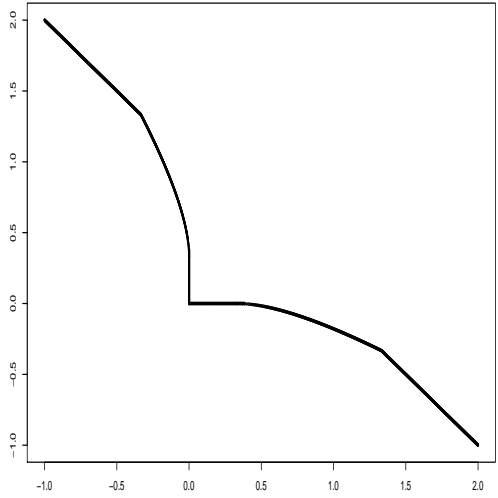}
    \caption{The lower left boundary of the set $\rhos(\bI_\kappa)$
      for $\kappa=1$ and $\alpha=0.75$.}
    \label{fig:1}
  \end{figure}
\end{example}

\section{Finite sets of admissible transactions}
\label{sec:finite-trans-sets}

We consider another special case when the selection risk measure of a
non-convex set can be calculated explicitly. Assume that possible
transactions are restricted to belong to a finite deterministic set
$M$ in $\R^d$, that is,
\begin{displaymath}
  \bX=C+M+\R_-^d.
\end{displaymath}
Let $\vecrisk$ have all components $\risk$ being the distortion risk
measure \eqref{eq:77} with distortion function $g$.  Since the
analytical calculation of $\rhos(\bX)$ is not feasible, it is possible
to use \eqref{eq:10} to arrive at the bound
\begin{displaymath}
  \rhos(C+M+\R_-^d)\supseteq \vecrisk(C)+\rhos(M+\R_-^d). 
\end{displaymath}
In the following we determine the last term on the right-hand side in
dimension $d=2$.

\begin{example}
  \label{ex:two-point}
  Consider the case of a two-point set $M$.  By translating, it is
  always possible to assume that $0\in M$.  If $M$ consists of two
  points $(0,0)$ and $(x,y)$ with $xy<0$, then $\rhos(\bX)$ is
  determined by the set of values $\vecrisk((x,y)\one_A)$ for all
  $A\in\salg$. Without loss of generality assume that $x>0$ and
  $y<0$. Since $\risk(\one_A)=-g(\beta)$ and
  $\risk(-\one_A)=1-g(1-\beta)=\tilde{g}(\beta)$ if $\P(A)=\beta$, we have
  \begin{displaymath}
    \rhos(M+\R_-^2)=\bigcup_{\beta\in[0,1]}
    \Big(-g(\beta)x,(g(1-\beta)-1)y\Big)+\R_+^2. 
  \end{displaymath}
\end{example}

\begin{example}
  Let $M=\{(x_1,y_1),(x_2,y_2),(x_3,y_3)\}$ consist of three points,
  and assume that $x_1<x_2=0<x_3$ and $y_1>y_2=0>y_3$. In this case,
  possible selections can be either two-points-selections of two of
  these three points (in this case the risk is calculated as in
  Example~\ref{ex:two-point}),
  and three point selection attaining all three points with positive
  probabilities $\alpha_1,\alpha_2,\alpha_3$ such that
  $\alpha_1+\alpha_2+\alpha_3=1$. 
  The risk of the three-point selection can be directly calculated, so
  that 
  \begin{displaymath}
    \rhos(M+\R_-^2)=\bigcup_{\alpha_1+\alpha_3\leq 1,\alpha_1,\alpha_3\geq0}
    \Big(-x_1\tilde{g}(\alpha_1)-x_3g(\alpha_3),
    -y_1g(\alpha_1)-y_3\tilde{g}(\alpha_3)\Big)+\R_+^2.
  \end{displaymath}
\end{example}



\newcommand{\noopsort}[1]{} \newcommand{\printfirst}[2]{#1}
  \newcommand{\singleletter}[1]{#1} \newcommand{\switchargs}[2]{#2#1}


\end{document}